\documentclass[12pt, journal, draftclsnofoot, onecolumn]{IEEEtran}
\usepackage{mathtools,enumitem,color,subfigure,bbm}
\usepackage{graphicx,tabularx,array}
\usepackage{hyperref}
\usepackage{amsthm,mathrsfs}
\usepackage{makecell}
\usepackage{multirow}
\usepackage[T1]{fontenc}
\usepackage{cite}

\usepackage[utf8]{inputenc} 
\usepackage[T1]{fontenc}
\usepackage{url}
\usepackage{ifthen}
\usepackage{enumitem}
\IEEEoverridecommandlockouts

\usepackage{color}
\usepackage{graphicx,tabularx,array,amsmath,amsthm,thmtools,caption}

\usepackage{mathtools}
\DeclarePairedDelimiter{\ceil}{\lceil}{\rceil}
\DeclarePairedDelimiter\floor{\lfloor}{\rfloor}
\usepackage{amsfonts}
\usepackage{bm}
\usepackage{bbm}
\usepackage{makecell}
\usepackage{multirow}
 \usepackage{amssymb}
\usepackage{txfonts}
\usepackage[T1]{fontenc}
\usepackage{tikz}
\usepackage[scr=dutchcal]{mathalfa}
\usepackage{ textcomp }
\usepackage{floatflt}
\usepackage{amsthm}

\usepackage{cite}
\DeclareMathAlphabet\mathbfcal{OMS}{cmsy}{b}{n}

\theoremstyle{definition}
\newtheorem{Theorem}{Theorem}
\newtheorem{Proposition}{Proposition}

\newtheorem{Example}{Example}
\newtheorem{Remark}{Remark}
\newtheorem{Definition}{Definition}

\begin{document}

%---------------------------------------------

%\begin{titlepage}%

%---------------------------------------------

%\begin{titlepage}
\title{Quantifying the Capacity Gains in Coarsely Quantized SISO Systems with Nonlinear Analog Operators}

\author{

\IEEEauthorblockN{ Farhad Shirani$^\dagger$ \thanks{This work was supported in part by NSF grant CCF-2132843.}, Hamidreza Aghasi$\ddagger$}
\IEEEauthorblockA{$^\dagger$Florida International University, $^\ddagger$ University of California, Irvine,
\\Email: fshirani@fiu.edu, haghasi@uci.edu}
}

\maketitle

 \begin{abstract}
The power consumption of high-speed, high-resolution analog to digital converters (ADCs) is a limiting factor in implementing large-bandwidth mm-wave communication systems. A mitigating solution, which has drawn considerable recent interest, is to use a few low-resolution ADCs at the receiver. While reducing the number and resolution of the ADCs decreases power consumption, it also leads to a reduction in channel capacity due to the information loss induced by coarse quantization.
This implies a rate-energy tradeoff governed by the number and resolution of ADCs. Recently, it was shown that given a fixed number of low-resolution ADCs, 
the application of practically implementable nonlinear analog operators, prior to sampling and quantization, may significantly reduce the aforementioned rate-loss. Building upon these observations, this work focuses on single-input single-output (SISO) communication scenarios, and i) characterizes capacity expressions  under various assumptions on the set of implementable nonlinear analog functions, ii) provides computational methods to calculate the channel capacity numerically, and iii) quantifies the gains due to the use of nonlinear operators in SISO receiver terminals. Furthermore, circuit-level simulations, using a 65 nm Bulk CMOS technology, are provided to show the implementability of the desired nonlinear operators in the analog domain. The power requirements of the proposed circuits are quantified for various analog operators. 
\end{abstract}


% \fi
\IEEEpeerreviewmaketitle
\vspace{-.06in}
\section{Introduction}
 In order to satisfy the ever-growing demand for higher data-rates, the fifth generation (5G) of wireless networks operate in a spectrum which includes frequencies above 6 GHz  especially the millimeter wave (mm-wave) bands. This allows for larger channel bandwidths
compared to earlier generation radio frequency (RF) systems which operate in lower frequency bands. The energy consumption of components such as analog to digital converters (ADCs) increases significantly with bandwidth  \cite{BR}.  For instance,
the power consumption of current commercial
high-speed ($\geq$ 20 GSample/s), high-resolution
(e.g. 8-12 bits) ADCs is around 500 mW per ADC \cite{zhang2018low}. In the standard ADC design, the power consumption is proportional to the number of quantization bins and hence grows exponentially in the number of output bits \cite{BR}.  As a result, one method which has been proposed to address high power consumption in mm-wave systems is to use a few {low-resolution ADCs}  at the receiver \cite{heath2016overview,nossek2006capacity,abbasISIT2018,khalili2020throughput,mollen2017achievable,jacobsson2017throughput}. The application of low-resolution ADCs poses fundamental questions in the design of receiver architectures, coding strategies, and capacity analysis of the resulting communication systems.
  
 \begin{figure}[t]
\centering 
\includegraphics[width=0.8\textwidth]{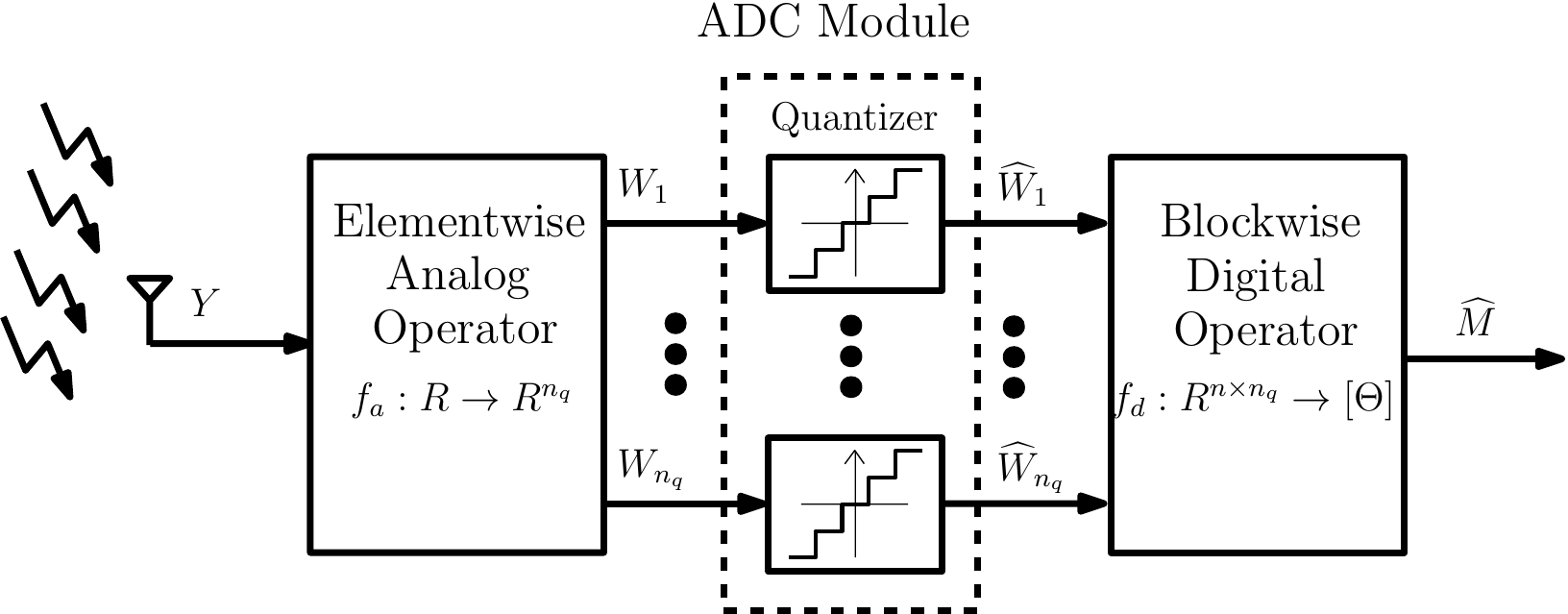}
\caption{The receiver architecture consists of an elementwise analog operator $f_a(\cdot)$, $n_q$ low-resolution ADCs, and a blockwise digital operator $f_d(\cdot)$ with blocklength $n$. $Y$ represents the received signal, $\widehat{M}$ is the message reconstruciton, and $[\Theta]$ is the message set.
}
\vspace{-.15in}
\label{fig:receiver}
\end{figure}

This work focuses on the receiver architectures and set of achievable rates in single-input single-output (SISO) systems equipped with low resolution ADCs. The setup has been considered extensively in prior works. An important initial result was due to the elegant approach proposed by Witsenahusen \cite{witsenhausen1980some}, which implies that, under peak power constraints, the capacity of a SISO system with one $K$-bit ADC   is achieved by a discrete input distribution with at most $K+1$ mass points. Later, this extended to SISO scenarios with average input power constraints \cite{singh2009limits}. For multiple-input multiple-output (MIMO) systems with low-resolution ADCs under peak power constraints, it was shown that the optimal input distribution has a finite discrete alphabet  \cite{dytso2018discrete}. Similarly, in multiterminal communications, for multiple-access channels (MAC) with a single antenna at each terminal and a single one-bit ADC at the receiver, input cardinality bounds were derived under peak power constraints \cite{rassouli2018gaussian}. It should be noted that although in these scenarios the input distribution has a discrete and finite alphabet with known cardinality, the optimization in the channel capacity expression is complex since each choice of input mass points and quantization thresholds yields a different set of channel transition probabilities. As a result, deriving analytical expressions for the channel capacity is challenging and often computational methods are proposed to evaluate the capacity, e.g. the cutting-plane algorithm \cite{huang2005characterization,singh2009limits}.

Recently, we considered MIMO communication systems with one-bit ADCs, and 
showed that the use of nonlinear analog operators, whose output is a polynomial function of their input, prior to sampling and quantization at the ADCs may significantly reduce the rate-loss due to coarse quantization
%; and derived capacity expressions under various assumptions on the channel signal-to-noise ratio and the set of implementable analog operators
\cite{shirani2022MIMO}. The receiver setup is shown in Fig. \ref{fig:receiver}. Furthermore, we introduced an analog circuit design  which produces a quadratic function of its input signal. The underlying idea in the circuit design is to leverage the nonlinearities of analog components to produce harmonics of the input signal, which are then extracted via frequency filtering techniques. 
It should be noted that the circuit complexity and power consumption increases with the degree of the desired polynomial. As a result, there are practical constraints on the degree of polynomials which are implementable under a given power budget. In this work, we build upon the observations in \cite{shirani2022MIMO} and investigate transceiver design and the resulting channel capacity in SISO systems. The main contributions of this work are summarized below:
\begin{itemize}[leftmargin=*]
    \item To characterize the high SNR SISO channel capacity as a function of the number of ADCs, $n_q$, number of output levels of each ADC, $\ell$, and maximum polynomial degree which is implementable using analog circuits, $\delta$. 
    \item To provide computational methods for finding the channel capacity and quantifying the gains due to nonlinear analog processing, and to provide explanations of how these gains change as  SNR, $n_q$, $\ell$, and $\delta$ are changed. 
    \item To provide circuit designs and associated performance simulations for implementing polynomials of degree up to four; and to evaluate their power consumption.
\end{itemize}

It should be noted that for MIMO scenarios with one-bit ADCs closed-form capacity expressions are derived in \cite{shirani2022MIMO} in terms of single-letter information measures, . However, these expressions involve optimization steps which may not be computable in general, or have high computational complexity. In contrast, in this work, we consider general low resolution ADCs --- as opposed to one-bit ADCs --- and provide computational methods to quantify the resulting channel capacity. This is an important step towards quantifying the gains due to nonlinear analog processing compared to beamforming architectures studied in prior works which use linear analog processing, and characterizing the tradeoffs between circuit complexity, power consumption, and channel capacity.

{\em Notation:}
%The random variable $\mathbbm{1}_{\mathcal{E}}$ is the indicator function of the event $\mathcal{E}$.
 The set $\{1,2,\cdots, n\}, n\in \mathbb{N}$ is represented by $[n]$. 
The  vector $(x_1,x_2,\hdots, x_n)$ is written as $x(1\!\!:\!\!n)$ and $x^n$, interchangeably. Similarly,  we interchange $x(i)$ and $x_i$. The vector $(x_k,x_{k+1},\cdots,x_n)$ is denoted by $x(k:n)$. We write $||\cdot||_2$ to denote the $L_2$-norm. An $n\times m$ matrix is written as $h(1\!\!:\!\!n,1\!\!:\!\!m)=[h_{i,j}]_{i,j\in [n]\times [m]}$, its $j$th column is $h(:,j), j\in [m]$, and its $i$th row is $h(i,:), i\in [n]$. We write $\mathbf{x}$ and $\mathbf{h}$ instead of $x(1\!\!:\!\!n)$ and $h(1\!\!:\!\!n,1\!\!:\!\!m)$, respectively, when the dimension is clear from context. Sets are denoted by calligraphic letters such as $\mathcal{X}$, families of sets by sans-serif letters such as $\mathsf{X}$, and collections of families of sets by $\mathscr{X}$. $\Phi$ represent the empty set. For the set $\mathcal{A}\subset \mathbb{R}^n$, the set $\partial\mathcal{A}_k$ denotes its boundary. $\mathbb{B}$ denotes the Borel $\sigma$-field. For the event $\mathcal{E}$, the variable $\mathbbm{1}(\mathcal{E})$ denotes the indicator of the event.

\section{System Model}
\label{sec:form}
We consider a SISO channel, whose input and output $(X, Y)\in \mathbb{R}^2$ are related through $
{Y}={h}{X}+{N}$, where  ${N}\in \mathbb{R}$ is a  Gaussian variable with unit variance and zero mean, and ${h}\in \mathbb{R}$ is the (fixed) channel gain coefficient. We assume that the transmitter and receiver have prefect knowledge of ${h}$, and the channel input has average power constraint $P$, i.e. $\mathbb{E}(X^2)\leq P$. Let the message $M$ be chosen uniformly from $[\Theta]$, where $\Theta \in \mathbb{N}$. The communication blocklength is $n\in \mathbb{N}$ and  the communication rate is $\frac{1}{n}\log{\Theta}$. The transmitter produces $e(M)=X^n$, where $e:[\Theta]\to \mathbb{R}^{n}$ is the encoding function. At the $i$th channel-use, the input $X(i), i\in [n]$ is transmitted and the receiver receives $Y(i)={h}X(i)+N(i)$. The receiver produces the message reconstruction $\widehat{M}=d(Y^n)$, where $d: \mathbb{R}^{n}\to[\Theta]$ is the decoding function. 

The choice of the decoding function $d(\cdot)$ is restricted by the limitations on the number of low-resolution threshold ADCs, $n_q\in \mathbb{N}$, the number of output levels of the ADCs, $\ell\in \mathbb{N}$, and the set of \textit{implementable nonlinear analog functions}: \[\mathcal{F}_a=\{f(x)=\sum_{i=0}^\delta{a_i}x^i, x\in \mathbb{R}| a_i \in \mathbb{R}, i=0,1,\cdots,\delta\},\]
\noindent 
which consists of all polynomials of degree at most $\delta\in \mathbb{N}$. The restriction to low-degree polynomial functions is due to limitations in analog circuit design, and the implementability of such functions is justified by the circuit designs and simulations provided in Section \ref{sec:cir}.

The receiver architecture, shown in Figure \ref{fig:receiver}, consists of:
\\i) A set of elementwise analog processing functions $f_{a,j}\in \mathcal{F}_a, j\in [n_q]$ operating on channel output $Y$ and producing the vector $W(1:n_q)$, where $W(j)=f_{a,j}(Y), j\in [n_q]$.
\\ii) A set of $n_q$ ADCs, each with $\ell$ output levels and  threshold vectors $t(j,1:\ell-1)\in \mathbb{R}^{\ell-1}, j\in [n_q]$ operating on the vector $W(1:n_q), i\in [n]$ and producing  $\widehat{W}(1:n_q)$, where
\[\widehat{W}(j)=
k  \quad \text{ if } \quad W(j)\in [t(j,k),t(j,k+1)], k\in [0,\ell-1],\]
where $j\in [n_q]$ 
and we have defined $t(j,0)\triangleq -\infty$ and  $t(j,\ell)\triangleq\infty$. We call $t(1\!\!:\!n_q,1\!\!:\!\ell-1)$ the \textit{threshold matrix}.
\\ iii) A digital processing module represented by $f_d:\{0,1,\cdots,\ell-1\}^{n\times n_q}\to [\Theta]$, operating on the block of ADC outputs after $n$-channel uses $\widehat{W}(1\!\!:\!\!n,1\!\!:\!\!n_q)$. 
%After the $i$th channel-use, the analog processing module processes the received signal $Y(i)$ in the analog domain and produces $W(i,j)=f_{a,j}(Y(i)), i\in [n], j\in [n_q]$. The output $W(i,1:n_q), i\in [n]$ is fed to the ADCs with threshold $t^{\ell-1}_j, j\in [n_q]$ which produce the discretized vector 
 After the $n$th channel-use, the digital processing module produces the message reconstruction $\widehat{M}=f_d(\widehat{W}(1\!:\!n,1\!:\!n_q))$. The communication system is characterized by $(P,{h}, n_q, \delta,\ell)$, and the transmission system by $(n,\Theta,e,f^{n_q}_a,t(1\!:\!n_q,1\!:\!\ell-1),f_d)$, where $f^{n_q}_a=(f_{a,1},f_{a,2},\cdots,f_{a,n_q})$ and $f_{a,j}, j\in [n_q]$ are polynomials with degree at most $\delta$, and $\mathbb{E}(e^2(I))\leq P$ for $I$ distributed uniformly on $[\Theta]
$. Achievability and probability of error are defined in the standard sense. The capacity maximized over all implementable analog functions is denoted by $C_{Q}(P,h,n_q, \delta,\ell)$.

\section{Communication Strategies and Achievable Rates}
\label{sec:results}
In this section, we investigate the SISO channel capacity for a given communication system parametrized by $(P,h,n_q,\delta,\ell)$.
\subsection{Preliminaries}
%One-bit ADCs and Quadratic Polynomials}
\label{sec:sen:III}
Let us consider the scenario where $n_q$ one-bit ADCs are used at the receiver, where $n_q>1$, and the set of implementable analog functions $\mathcal{F}_a$ is restricted to quadratic functions, i.e., $\delta=\ell=2$. In \cite[Th. 4]{shirani2022MIMO}, the high SNR capacity was derived for MIMO systems with one-bit ADCs. The result implies that the high SNR SISO capacity is equal to $1+\log{n_q}$ bit/channel-use, and is strictly greater than $\log{(1+n_q)}$ bit/channel-use, which is the hybrid beamforming capacity where linear analog processing is used.
The proof relies on a geometric argument. To elaborate, it was argued that the number of messages transmitted per channel-use is equal to the number of partition regions of the output space imposed by the ADC quantization process. For one-bit threshold ADCs with linear analog processing, the number of partition regions is equal to $n_q+1$, hence the high SNR capacity is $\log{(1+n_q)}$, whereas when quadratic functions are used for analog processing, the maximum number of partition regions is equal to $2n_q$, yeilding the capacity of $1+\log{n_q}$. The latter statement is proved by counting the maximum number of partition regions imposed on the one-dimensional manifold $\{(Y,Y^2)| Y\in \mathbb{R}\}$ in partitions of $\mathbb{R}^2$ by $n_q$ lines. This geometric argument does not extend naturally if ADCs with more than two output levels and higher degree polynomial functions are used, i.e. $\delta>2$ and $\ell>2$. In the sequel, we provide an alternative proof of \cite[Th. 4]{shirani2022MIMO} for scenarios with $\delta=\ell=2$. We build upon this to derive capacity expressions for $\delta,\ell\in \mathbb{N}$. To this end, we first introduce some useful terminology and preliminary results. 

The quantization process at the receiver is modeled by two sets of functions. The analog processing functions $f_{a,j}(\cdot), j\in [n_q]$ and ADC threshold matrix $t(1:n_q,1:\ell-1)$. 
\begin{Definition}[\textbf{Quantizer}] A quantizer $Q:\mathbb{R}\to [\ell]^{n_q}$ characterized by the tuple $(\ell,\delta,n_q,f_{a}^{n_q}(\cdot), t(1\!:\!\!n_q,1\!:\!\ell\!-\!1))$ is defined as $Q(\cdot)\triangleq (Q_{1}(\cdot),Q_{2}(\cdot),\cdots,Q_{n_q}(\cdot))$, where $Q_{j}(y)\triangleq k$ if $ f_{a,j}(y)\in [t(j,k),t(j,k+1)], j\in [n_q]$, the functions $f_{a,j}(\cdot)$, $j\!\in\! [n_q]$ are polynomials of degree at most $\delta$, and $t(1\!:\!n_q,\!1\!:\!\ell\!-\!1)\!\in \!\mathbb{R}^{n_q\times {(\ell-1)}}$. The associated partition of $Q(\cdot)$ is: 
\begin{align*}
\mathsf{P}=\{\mathcal{P}_{\mathbf{i}}, \mathbf{i}\in [\ell]^{n_q}\}- \Phi, \text{ where } \mathcal{P}_\mathbf{i}= \{y\in\mathbb{R}| Q(y)= \mathbf{i}\}, \mathbf{i}\in [\ell]^{n_q}.
\end{align*}
\label{def:quant}
\end{Definition}
\vspace{-.2in}
 For a quantizer $Q(\cdot)$, we call $y\in \mathbb{R}$ a \textit{point of transition} if the value of $Q(\cdot)$ changes at input $y$, i.e. if it is a point of discontinuity of $Q(\cdot)$. Let $r$ be a point of transition of $Q(\cdot)$. Then, there must exist output vectors $\mathbf{c}\neq \mathbf{c}'$ and $\epsilon>0$ such that $Q(y)=\mathbf{c}, y\in (r-\epsilon,r)$ and $Q(y)=\mathbf{c}', y\in (r,r+\epsilon)$. So, there exists $j\in [n_q]$ and $k\in [\ell-1]$ such that $f_{a,j}(y)<t(j,k), y\in (r-\epsilon,r)$ and $f_{a,j}(y)\geq t(j,k), r\in (r,r+\epsilon)$, or vice versa; so that $r$ is a root of the polynomial $f_{a,k,j}(y)\triangleq f_{a,j}(y)-t(j,k)$.  Let $r_1,r_2,\cdots,r_{(\ell-1) \delta n_q}$ be the sequence of roots of polynomials $f_{a,j,k}(\cdot), j\in [n_q], k\in [\ell-1]$ (including repeated roots), written in non-decreasing order, and let $\mathcal{C}=(\mathbf{c}_0,\mathbf{c}_1,\cdots, \mathbf{c}_{(\ell-1) \delta n_q})$ be the corresponding quantizer outputs, i.e. $\mathbf{c}_{i-1}= \lim_{y\to r_i^-}Q(y), i\in [(\ell-1) \delta n_q]$ and $\mathbf{c}_{(\ell-1) \delta n_q}=\lim_{y\to\infty}Q(y)$. We call $\mathcal{C}$ the \textit{code} associated with the quantizer and it plays an important role in the analysis provided in the sequel. Note that the associated code is an ordered set of vectors.
The size of the code $|\mathcal{C}|$ is defined as the number of unique vectors in $\mathcal{C}$. Each $\mathbf{c}_i= (c_{i,1},c_{i,2},\cdots,c_{i,n_q}), i\in \{0,1,\cdots,(\ell-1)\delta n_q\}$ is called a codeword. For a fixed $j\in [n_q]$, the transition count of position $j$ is the number of codeword indices where the value of the $j$th element changes, and it is denoted by $\kappa_j$,i.e., $\kappa_j\triangleq \sum_{k=1}^{(\ell-1)\gamma n_q}\mathbbm{1}(c_{i_k-1,j}\neq c_{i_k,j})$.
It is straightforward to see that $|\mathsf{P}|=|\mathcal{C}|$ since both cardinalities are equal to the number of unique outputs the quantizer produces. 
The following example clarifies the definitions given above. 
 
%  Note that the points of transition are the boundary points of partition regions $\mathcal{P}_{\mathbf{i}}, \mathbf{i}\in [\ell]^{n_q}$.
% Let us consider the following experiment. We start with an asymptotically large negative output value, i.e. $y\to -\infty$, and gradually increase the value of $y$ and investigate the points of transition of the quantizer. Let  $\lim_{y\to -\infty} Q(y)= \mathbf{c}_0=(c_{0,1},c_{0,2},\cdots,c_{0,n_q})\in [\ell]^{n_q}$, and let $r_1\in \mathbb{R}$ be the first point of transition. That is $Q(y)=\mathbf{c}_0, y\in (-\infty,r_1)$ and $\lim_{y\to r_1^+}Q(y)=\mathbf{c}_1\in [\ell]^{n_q}$, where $\mathbf{c}_0\neq \mathbf{c}_1$. Then, there must exist $j\in [n_q]$ and $k\in [\ell-1]$ such that $\lim_{y\to r_1^-}f_{a,j}(y)<t(j,k)$ and $\lim_{y\to r_1^+}f_{a,j}(y)\geq t(j,k)$, or vice versa; so that $r_1$ is a root of the polynomial $f_{a,k,j}(y)\triangleq f_{a,j}(y)-t(j,k)$.
 \begin{figure}[t]
\centering 
\includegraphics[width=0.8\textwidth]{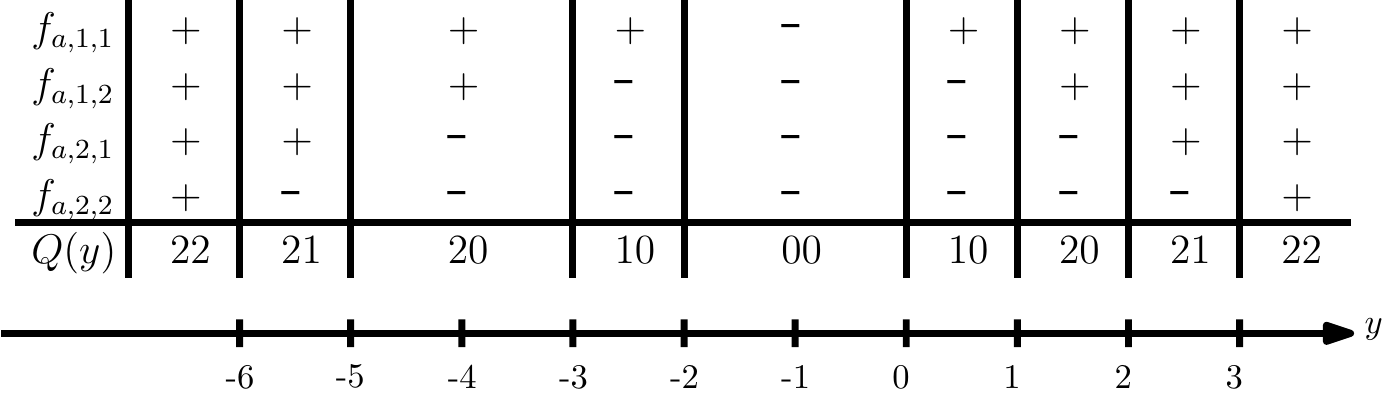}
\caption{The quantizer outputs in Example \ref{Ex:1}. The first four rows show the sign of the function $f_{a,j,k}, j,k\in \{1,2\}$ for the values of $y$ within each interval. The last row shows the quantizer output in that interval.
}
\vspace{-.15in}
\label{fig:code}
\end{figure}

\begin{Example}[\textbf{Associated Code}]
\label{Ex:1}
Let $n_q=\delta=2$ and $\ell=3$ and consider a quantizer characterized by polynomials
$f_{a,1}(y)=y^2+2y$ and $f_{a,2}(y)= y^2+3y, y\in \mathbb{R}$, and thresholds
\begin{align*}
& t(1,1)= 3, \quad t(1,2)=0, \quad t(2,1)= 10, \quad t(2,2)= 18,
\end{align*}
We have:
\begin{align*}
    & f_{a,1,1}(y)= y^2+2y-3, \quad f_{a,1,2}(y)= y^2+2y\\
    & f_{a,2,1}(y)=y^2+3y-10, \quad  f_{a,2,2}(y)= y^2+3y-18.
\end{align*}
The ordered root sequence is $(r_1,r_2,\cdots,r_8)=
(-6,$ $-5,-3,-2,0,1,2,3)$. The 
associated partition is:
\begin{align*}
   & \mathsf{P}= \Big\{[-\infty,-6), (-6,-5),(-5,-3),
 (-3,-2),(-2,0),
 \\&\qquad (0,1),(1,2),(2,3),(3,\infty)\Big\}.
\end{align*}
The associated code is given by $22,21,20,10,00,10,20,21,22$. This is shown in Figure \ref{fig:code}. The size of the code is $|\mathcal{C}|=5$. The high SNR capacity of a SISO channel using this quantizer at the receiver is $\log{|\mathsf{P}|}=\log{|\mathcal{C}|}=\log{5}$.
\end{Example}
\label{sec:prelim}

\subsection{SISO Systems with One-bit ADCs and Quadratic Functions}
To illustrate the usefulness of the notion of associated code of a quantizer, introduced in the prequel, let us prove the high SNR SISO capacity result for $\ell=\delta=2$ given in \cite[Th. 4]{shirani2022MIMO} using the framework introduced in Section \ref{sec:prelim}.
\begin{Proposition}
\label{Prop:1}
Let $h\in \mathbb{R}$ and $n_q>1$. Then, 
\begin{align*}
 \lim_{P\to \infty}   C_{Q}(P,h,n_q,2)= 1+\log{n_q}. 
\end{align*}
\end{Proposition}
\begin{proof}
For a given quantizer, the high SNR achievable rate is equal to $\log{|\mathsf{P}|}= \log{|\mathcal{C}|}$. So, finding the capacity is equivalent to finding the maximum $|\mathcal{C}|$ over all choices of $Q(\cdot)$.
First, let us prove the converse result. Note that $|\mathcal{C}|\leq 2n_q$ since $\mathbf{c}_0=\mathbf{c}_{2n_q}$. The reason is that for the quadratic function $f_{a,j}(\cdot),j\in [n_q]$, we have $\lim_{y\to \infty}f_{a,j}(y)=\lim_{y\to -\infty}f_{a,j}(y)\in \{-\infty,\infty\} $. So, \[c_{0,j}=\lim_{y\to -\infty} \mathbbm{1}(f_{a,j}-t_j>0)=\lim_{y\to \infty} \mathbbm{1}(f_{a,j}-t_j>0)=c_{2n_q,j}.\]
As a result, $\log{|\mathcal{C}|}\leq 1+\log{n_q}$. Next, we prove achievability.
Let $t_j=0, j\in [n_q]$ and $f_{a,j}(y)\triangleq -(y+n_q+1-j)(y-j), j\in [n_q]$. Then, $(r_1,r_2,\cdots,r_{2n_q})\!=\!(-n_q,-n_q\!+\!1,\cdots,-1,1,2,\cdots, n_q)$ and \begin{align*}
    c(i,j)=
    \begin{cases}
    \mathbbm{1}(j\leq i)\qquad& \text{if }\quad  i\leq n_q,\\
    1-\mathbbm{1}(j\leq i-n_q) & \text{otherwise}.
    \end{cases}
\end{align*}
For instance for $n_q=3$, we have $\mathcal{C}=(000,001,011,$ $111,110,100,000)$. It is straightforward to see that the only repeated codewords are $\mathbf{c}_0$ and $\mathbf{c}_{2n_q}$. Hence, $|\mathcal{C}|=2n_q$, and $\log{|\mathcal{C}|}=1+\log{n_q}$ is achievable. 
\end{proof}

We can further provide the following computable expression for the capacity under general assumptions on channel SNR. 

\begin{Theorem}
\label{th:1}
Consider a SISO system parametrized by $(P,h,n_q,\delta,\ell)$, where $P>0, h\in \mathbb{R}, n_q>1$, and $\delta=\ell=2$. Then, the capacity is given by:
\begin{align}
\label{eq:th:1}
    C_Q(P,h,n_q,\delta,\ell)=\sup_{\mathbf{x}\in \mathbb{R}^{2n_q+1}} \sup_{P_{X}\in \mathcal{P}_{\mathbf{x}}} \sup_{\mathbf{t}\in \mathbb{R}^{2n_q}} I(X;\widehat{Y}),
\end{align}
where $\widehat{Y}= Q(hX+N)$, $N$ is a zero-mean Gaussian variable with unit variance, $\mathcal{P}_{\mathbf{x}}$ is the probability simplex on alphabet $\{x_1,x_2,\cdots,x_{2n_q+1}\}$,
and $Q(y)=k$ if $y\in [t_{k},t_{k+1}], k\in \{1,\cdots,2n_q]$ and $Q(y)=0$ if $y>t_{2n_q}$ or $y<t_{1}$.
\end{Theorem}

\begin{proof}
We provide an outline of the proof. First, we prove that the input alphabet has at most $2n_q+1$ mass points.
Based on the proof of Proposition \ref{Prop:1}, the channel output can take at most $2n_q$ values. Let the quantized channel output be denoted by $\widehat{Y}$.
Since the conditional measure $P_{\widehat{Y}|X}(\cdot|x), x\in\mathbb{R}$ is continuous in $x$, and $\lim_{x\to \infty} P_{\widehat{Y}|X}(\mathcal{A}|x) = \mathbbm{1}(\hat{y}\in \mathcal{A}), \mathcal{A}\in \mathbb{B}$ for some fixed $\hat{y}$, the conditions in the proof of \cite[Prop. 1]{singh2009limits} hold, and the optimal input distribution has bounded support. Then, using the extension of Witsenhausen's result \cite{witsenhausen1980some} given in \cite[Prop. 2]{singh2009limits}, the optimal input distribution is discrete and takes at most ${2n_q+1}$ values. This completes the proof of converse. To prove achievability, it suffices to show that one can choose the set of quadratic functions $f_{a,j}(\cdot)$ and quantization thresholds such that the resulting quantizer operates as described in the theorem statement.
Let $\mathbf{t}^*$ be the optimal quantizer thresholds in \eqref{eq:th:1}. Let $r_1,r_2,\cdots,r_{2n_q}$ be the elements of $\mathbf{t}^*$ written in non-decreasing order. Define a quantizer with associated polynomials $f_{a,j}(y)\triangleq  -(y-r_{j})(y-r_{n_q+j})$ and zero threshold vector. Then, similar to the proof of Proposition \ref{Prop:1}, the quantization rule gives distinct outputs for $y\in [r_{k},r_{k+1}], k\in \{1,\cdots,2n_q]$ and $y\in [r_{2n_q},\infty) \cup [-\infty,r_{1}]
$ as desired. \end{proof}

\begin{Remark}
The capacity expression in Equation \ref{eq:th:1} can be computed numerically, e.g. using the cutting plane algorithm \cite{huang2005characterization,singh2009limits}, or the extension of Blahut-Arimoto algorithm in \cite{kobayashi2018joint}. This is investigated in Section \ref{sec:num}.
\end{Remark}
\subsection{Low-resolution ADCs and Low-degree Polynomials}
We wish to extend our analysis to SISO systems with $\delta,\ell>2$.
Towards this, 
the following proposition states several useful properties for the code associated with a quantizer $Q(\cdot)$. These are straightforward extensions of the properties shown in the proof of Theorem \ref{th:1} and  their proof is omitted for brevity.
\begin{Proposition}[\textbf{Properties of the Associated Code}] 
\label{Prop:2}
Consider a quantizer $Q(\cdot)$ with threshold matrix $t(1:n_q,1:\ell-1)$ and associated polynomials $f_{a,j}(\cdot), j\in n_q$, such that 
 $f_{a,j,k}(\cdot)\triangleq f_{a,j}(\cdot)- t(j,k), j\in [n_q], k\in [\ell-1]$ do not have repeated roots. The associated code $\mathcal{C}$ satisfies the following:
\begin{enumerate}[leftmargin=*]
\item The number of codewords in $\mathcal{C}$ is equal to $\gamma\triangleq (\ell-1)\delta n_q$, i.e. $\mathcal{C}=(\mathbf{c}_0, \mathbf{c}_1,\cdots, \mathbf{c}_{\gamma-1})$.
    \item All elements of the first codeword $\mathbf{c}_0$ are either equal to $\ell-1$ or equal to $0$, i.e. $c_{i,0}=0, i\in \{0,1,\cdots,\gamma-1\}$ or $c_{i,0}=\ell, i\in \{0,1,\cdots,\gamma-1\}$.
    \item Consecutive codewords differ in only one position, and their $L_1$ distance is equal to one, i.e. $\sum_{j=1}^{n_q}|c_{i,j}-c_{i+1,j}|=1, i\in \{0,1,\cdots,\gamma-1\}$.
    \item  The transition count at every position is $\kappa_j= \frac{\gamma}{n_q}= (\ell-1)\delta, j\in [n_q]$.
    \item Let $i_1,i_2,\cdots, i_{\kappa}$ be the non-decreasingly ordered indices of codewords where the $j$th element has value-transitions. Then, the sequence $(c_{i_1,j},c_{i_2,j},\cdots,c_{i_\kappa,j})$ is periodic, in each period it takes all values between $0$ and $\ell-1$, and  $|c_{i_k,j}-c_{i_{k+1},j}|=1, k\in [\kappa-1]$ holds. Furthermore, $c_{i_1,j}\in \{0,\ell-1\}$.
    \item If $\delta$ is even, then $|\mathcal{C}|\leq min(\ell^{n_q}, (\ell-1)\delta n_q)$ and if $\delta$ is odd, then $|\mathcal{C}|\leq min(\ell^{n_q}, (\ell-1)\delta n_q+1)$
    
\end{enumerate}

\end{Proposition}

Next, we study the capacity region for SISO systems when $\ell=2$ and $\delta$ is even. First we prove two useful propositions. The first one proves that given an ordered set $\mathcal{C}$ satisfying the properties in Proposition \ref{Prop:2}, one can always construct a quantizer whose associated code is equal to $\mathcal{C}$. The second proposition provides conditions under which there exists a code satisfying the properties in Proposition \ref{Prop:2}. The proof ideas follow techniques used in study of balanced and locally balanced gray codes \cite{bhat1996balanced,bykov2016locally}. 
Combining the two results allows us to characterize the necessary and sufficient conditions for existence of quantizers with desirable properties.

\begin{Proposition}[\textbf{Quantizer Construction}]
\label{Prop:3}
Let $\ell=2, n_q\in \mathbb{N}$ and $\delta$ be an even number. 
Given an ordered set $\mathcal{C}\subset \{0,1\}^{n_q}$ satisfying properties 1)-5) in Proposition \ref{Prop:2}, and a sequence of non-decreasing real numbers $r_1,r_2,\cdots, r_{\gamma}$, where $\gamma=\delta n_q$. There exists a quantizer $Q(\cdot)$ with zero threshold vector and associated polynomials $f_{a,j}(\cdot), j\in [n_q]$ such that its associated code is $\mathcal{C}$, and $r_1,r_2,\cdots, r_{\gamma}$ is the non-decreasing sequence of roots of its associated polynomials $f_{a,j}(\cdot),j\in [n_q]$.
\end{Proposition}

\begin{proof}
Without loss of generality, let us assume that $\mathbf{c}_0$ is the all-zero sequence.  
Let $\gamma$ be the number of codewords in $\mathcal{C}$. Note that in general $\gamma$ may be larger than $|\mathcal{C}|$ since there might be repeated codeword sequences.
Let $t_1,t_2,\cdots,t_{\gamma-1}$ be the transition sequence of $\mathcal{C}$. That is, $t_k, k\in \{1,\dots,\gamma-1\}$ is the bit position which is different between $\mathbf{c}_{k-1}$ and $\mathbf{c}_{k}$. Consider a  quantizer $Q(\cdot)$ with zero threshold and associated polynomials $f_{a,j}(y)\triangleq -\prod_{k: t_k=j}(y-r_k), j\in [n_q]$. Then, $r_1,r_2,\cdots, r_{\gamma}$ are the non-decreasing sequence of roots of $f_{a,j}(\cdot), j\in [n_q]$, and the associated code of the quantizer $Q(\cdot)$ is $\mathcal{C}$ as desired.
\end{proof}
\begin{Proposition}(\textbf{Code Construction})
Let $\ell=2$, $n_q\in \mathbb{N}$, and $\kappa_1$, $\kappa_2,\cdots,\kappa_{n_q}$ be even numbers such that $|\kappa_j-\kappa_{j'}|\leq 2, j,j'\in [n_q]$. Then, there exists a code $\mathcal{C}$ with transition count at position j equal to $\kappa_j, j\in [n_q]$ satisfying properties 1), 2), 3), and 5) in Proposition \ref{Prop:2} such that $|\mathcal{C}|=\min\{2^{n_q}, \sum_{j=1}^{n_q}\delta_j\}$. Particularly, if $\kappa_j=\delta, j\in [n_q]$, then there exists $\mathcal{C}$ with $|\mathcal{C}|= \min\{2^{n_q},\delta n_q\}$ satisfying properties 1)-5) in Proposition \ref{Prop:2}.
\label{Prop:4}
\end{Proposition}
\begin{proof}
We provide an outline of the proof.  Let us consider the following cases:
\\\textbf{Case 1:} $\sum_{j=1}^{n_q }\kappa_j\geq 2^{n_q}$
\\In this case, one can use a balanced Gray code \cite{bhat1996balanced} to construct $\mathcal{C}$. A balanced Gray code is a (binary) code where consecutive codewords have Hamming distance equal to one, and each of the bit positions changes value either $2\floor{\frac{2^{n_q}}{2n_q}}$ times or $2\ceil{\frac{2^{n_q}}{2n_q}}$ times. If $\min_{j\in [n_q]} \kappa_j\geq 2\ceil{\frac{2^{n_q}}{2n_q}}$ the proof is complete as one can concatenate the balanced gray code with a series of additional repeated codewords to satisfy the transition counts, and since the balanced gray code is a subcode of the resulting code, we have $|\mathcal{C}|=2^{n_q}$. Otherwise,  there exists $j\in [n_q]$ such that $\kappa_j< 2\ceil{\frac{2^{n_q}}{2n_q}}$.
In this case, without loss of generality, let us assume that $\kappa_1\leq \kappa_2,\cdots \leq \kappa_{n_q}$. Note that since $|\kappa_j-\kappa_j'|\leq 2, j,j'\in [n_q]$ and $\kappa_j, j\in [n_q]$ are even, there is at most one $j^*\in [n_q]$ such that $\kappa_{j^*}\leq \kappa_{j^*+1}$. Let $\kappa'_1,\kappa'_2,\cdots,\kappa'_{n_q}$ be the transition count sequence of a balanced gray code $\mathcal{C}'$ written in non-decreasing order. 
Note that $2\ceil{\frac{2^{n_q}}{2n_q}}-2\floor{\frac{2^{n_q}}{2n_q}}=2$. Hence, similar to the above argument, there can only be one $j'\in [n_q]$ for  which   $\kappa_{j'}\leq \kappa_{j'+1}$. 
Since $\sum_{j=1}^{n_q} \kappa_j\geq  2^{n_q}= \sum_{j=1}^{n_q} \kappa'_j$, we must have $j^*\leq j'$. So, the balanced gray code can be used as a subcode similar to the previous case by correctly ordering the bit positions to match the order of $\kappa_j, j\in [n_q]$. This completes the proof. 
\\\textbf{Case 2:} $\sum_{j=1}^{n_q}\kappa_j< 2^{n_q}$
\\The proof is based on techniques used in the construciton of balanced Gray codes \cite{bhat1996balanced}.
We prove the result by induction on $n_q$.
The proof for $n_q=1,2$ is straightforward and follows by construction of length-one and length-two sequences. For $n_q>2$, Assume that the result holds for all $n'_q\leq n_q$. Without loss of generality, assume that $\kappa_1\leq \kappa_2,\leq \cdots \leq \kappa_{n_q}$. The proof considers four sub-cases as follows.
\\\textbf{Case 2.i:} $\sum_{j=3}^{n_q}\kappa_j\in [0,2^{n_q-2}]$
\\In this case, by the induction assumption, there exists $\mathcal{C}'$,  a code with codewords of length $n_q-2$, whose transition sequence is $\kappa_3,\kappa_4,\cdots,\kappa_{n_q}$, and $|\mathcal{C}'|= \sum_{j=3}^{n_q}\kappa_j$. We construct $\mathcal{C}$ from $\mathcal{C}'$ as follows. Let $\mathbf{c}_{0}=(0,0,\mathbf{c}'_0)$, $\mathbf{c}_{1}=(0,1,\mathbf{c}'_0)$, $\mathbf{c}_{2}=(1,1,\mathbf{c}'_0)$,
$\mathbf{c}_{3}=(1,0,\mathbf{c}'_0)$,
$\mathbf{c}_{4}=(1,0,\mathbf{c}'_1)$,
$\mathbf{c}_{5}=(0,0,\mathbf{c}'_1)$, $\mathbf{c}_{6}=(0,1,\mathbf{c}'_1)$,
$\mathbf{c}_{7}=(1,1,\mathbf{c}'_1)$,$\cdots$. This resembles the procedure for constructing balanced gray codes \cite{bhat1996balanced}. We continue concatenating the first two bits of each codeword in $\mathcal{C}$ to the codewords in $\mathcal{C}'$ using the procedure described above until $\kappa_1$ transitions for position 1 and $\kappa_2$ transitions for position 2 have taken place. Note that this is always possible since i) for each two codewords in $\mathcal{C}'$, we `spend' two transitions of each of the first and second positions in $\mathcal{C}$ to produce four new codewords, ii) $\kappa_2-\kappa_1\leq 2$, and iii) $\kappa_2\leq \sum_{j=3}^{n_q}\kappa_j$, where the latter condition ensures that we do not run out of codewords in $\mathcal{C}'$ before the necessary transitions in positions 1 and 2 are completed. After $\kappa_2+1$ codewords, the transitions in positions 1 and 2 are completed, and the last produced codeword is $(0,0,\mathbf{c}'_{\kappa_2+1})$  since $\kappa_1$ and $\kappa_2$ are both even. To complete the code $\mathcal{C}$, we add $(0,0,\mathbf{c}'_{i}), i\in [\kappa_2+2,\sum_{j=3}^{n_q}\kappa_j]$. Then, by construction, we have $|\mathcal{C}|=|\mathcal{C}'|+\kappa_1+\kappa_2=\sum_{j=1}^{n_q}\kappa_j$ and the code satisfied Properties 1), 2), 3), and 5) in Proposition \ref{Prop:2}. 
\\\textbf{Case 2.ii:}$\sum_{j=3}^{n_q}\kappa_j\in [2^{n_q-2}, 2^{n_q-1}]$
\\ Similar to the previous case, let $\mathcal{C}'$ be a balanced gray code with codeword length $n_q-2$ and transition counts $\kappa'_1\leq \kappa'_2\leq \cdots\leq \kappa'_{n_q-2}$. Define $\kappa''_{j}=\kappa_j- \kappa'_{j+2}, j\in \{3,4,\cdots,n_q\}$. Note that $\kappa''_j$ satisfy the conditions on transition counts in the proposition statement, and hence by the induction assumption, there exists a code $\mathcal{C}''$ with transition counts $\kappa''_j, j\in [n_q-2]$. The proof is completed by appropriately concatenating $\mathcal{C}'$ and $\mathcal{C}''$ to construct $\mathcal{C}$. Let $\gamma''$ be the number of codewords in $\mathcal{C}''$ and define $\mathbf{c}_i=(0,0,\mathbf{c}''_i), i\in [\gamma'']$, $\mathbf{c}_{\gamma''+1}=(0,1,\mathbf{c}''_{\gamma''})$,  $\mathbf{c}_{\gamma''+2}=(1,1,\mathbf{c}''_{\gamma''})$, $\mathbf{c}_{\gamma''+3}=(1,0,\mathbf{c}''_{\gamma''})$, $\mathbf{c}_{\gamma''+4}=(1,0,\mathbf{c}'_{1})$,$\cdots$. Similar to the previous case, it is straightforward to show that this procedure yields a code $\mathcal{C}$ with the desired transition sequence. 

The proof for the two subcases where $\sum_{j=3}^{n_q}\kappa_j\in [2^{n_q-1},3 \times 2^{n_q-2}]$ and $\sum_{j=3}^{n_q}\kappa_j\in [3\times 2^{n_q-1}, \times 2^{n_q-1}]$ is similar and is ommited for brevity.
\end{proof}
Using Propositions \ref{Prop:3} and \ref{Prop:4}, we characterize the SISO capacity for $\ell=2$ and even-valued $\delta$.
\begin{Theorem}
\label{th:2}
Consider a SISO system parametrized by $(P,h,n_q,\delta,\ell)$, where $P>0, h\in \mathbb{R}, n_q\in \mathbb{N}$, $\delta\in \{2,4,6,\cdots\}$, and $\ell=2$. Then, the capacity is given by:
\begin{align}
\label{eq:th:2}
    C_Q(P,h,n_q,\delta,\ell)=\sup_{\mathbf{x}\in \mathbb{R}^{ \Gamma}} \sup_{P_{X}\in \mathcal{P}_{\mathbf{x}}} \sup_{\mathbf{t}\in \mathbb{R}^{\Gamma-1}} I(X;\widehat{Y}),
\end{align}
where $\Gamma\triangleq \min(2^{n_q}, \delta n_q)$,
$\widehat{Y}= Q(hX+N)$, $N$ is a zero-mean Gaussian variable with unit variance and, $\mathcal{P}_{\mathbf{x}}$ is the probability simplex on alphabet $\{x_1,x_2,\cdots,x_{\Gamma}\}$,
and $Q(y)=k$ if $y\in [t_{k},t_{k+1}], k\in \{1,\cdots,\Gamma-1]$ and $Q(y)=0$ if $y>t_{\Gamma-1}$ or $y<t_{1}$.
\end{Theorem}

The proof follows by similar arguments as in the proof of Theorem \ref{th:1}. The converse follows from Proposition \ref{Prop:2} Item 4). Achievability follows from Proposition \ref{Prop:4}.

Furthermore, using property 6) in Proposition \ref{Prop:2} along with the proof of Theorem \ref{th:1}, we derive the following upper and lower bounds to the case when $\delta$ is an odd number. 

\begin{Theorem}
Consider a SISO system parametrized by $(P,h,n_q,\delta,\ell)$, where $P>0, h\in \mathbb{R}, n_q\in \mathbb{N}$, $\delta\in \{1,3,5,\cdots\}$, and $\ell=2$. Then, the capacity satisfies:
\begin{align}
\label{eq:th:3}
 \sup_{\mathbf{x}\in \mathbb{R}^{\Gamma}} \sup_{P_{X}\in \mathcal{P}_{\mathbf{x}}} \sup_{\mathbf{t}\in \mathbb{R}^{\Gamma-1}} I(X;\widehat{Y})  & \leq C_Q(P,h,n_q,\delta,\ell)
 \\&\nonumber \leq \sup_{\mathbf{x}\in \mathbb{R}^{\Gamma'}} \sup_{P_{X}\in \mathcal{P}_{\mathbf{x}}} \sup_{\mathbf{t}\in \mathbb{R}^{\Gamma'-1}} I(X;\widehat{Y}),
\end{align}
where $\Gamma\triangleq \min(2^{n_q}, \delta n_q)$ and  $\Gamma'\triangleq \min(2^{n_q}, \delta n_q+1)$ .
\end{Theorem}

Lastly, for for scenarios with $\ell>2$ the following theorem characterizes the channel capacity. The proof follows from Propositions \ref{Prop:2} and \ref{Prop:4} similar to the arguments given in the proof of Theorem \ref{th:1}.

\begin{Theorem}
\label{th:4}
Consider a SISO system parametrized by $(P,h,n_q,\delta,\ell)$, where $P>0, h\in \mathbb{R}, n_q\in \mathbb{N}$, and $\ell,\delta\in \mathbb{N}$. Let $\Gamma$ be the maximum size of codes satisfying condition 1)-5) in Proposition \ref{Prop:2}.  Then, 
\begin{align}
\label{eq:th:4}
 C_Q(P,h,n_q,\delta,\ell) =\sup_{\mathbf{x}\in \mathbb{R}^{\Gamma}} \sup_{P_{X}\in \mathcal{P}_{\mathbf{x}}} \sup_{\mathbf{t}\in \mathbb{R}^{\Gamma-1}} I(X;\widehat{Y}).
\end{align}
\end{Theorem}

Optimizing \eqref{eq:th:4} requires calculating $\Gamma$. The number of codes satisfying conditions 1)-5) in Proposition \ref{Prop:2} is bounded from above by ${(\ell-1)\delta n_q\choose (\ell-1)\delta, (\ell-1)\delta , \cdots, (\ell-1)\delta }$. 
For SISO systems with a few low resolution ADCs and low degree polynomials (small $\ell,n_q$ and $\delta$), one can directly calculate $\Gamma$ by optimizing over the set of such codes. 

% \begin{Proposition}
% Let $\ell=2$ and $n_q,\delta \in \mathbb{N}$. Then for any increasing sequence $r_1,r_2,\cdots, r_{\delta n_q}\in \mathbb{R}^{\delta n_q}$, there exists polynomials $f_{a,j},j\in [n_q]$ and thresholds $t(1:n_q)$ such that the ordered root sequence of $f_{a,j}(\cdot)$ is given by $r_1,r_2,\cdots,r_{\delta n_q}$, and $|\mathcal{C}|\geq  min(2^{n_q}, \delta n_q)$.
% \end{Proposition}

\section{Numerical Analysis of SISO Channel Capacity}
\label{sec:num}
In this section, we provide a numerical analysis of the capacity bounds derived in Section \ref{sec:results} and evaluate the gains due to the use of nonlinear analog components in the receiver terminal. In particular, we compute inner-bounds to the capacity expression in Section \ref{sec:results} using the extension of the Blahut-Arimoto algorithm to discrete memoryless channels with input cost constraints given in \cite{kobayashi2018joint} to find the best input distribution, and then we conduct a brute-force search over all possible symmetric threshold vectors, where a vector $\mathbf{t}$ is symmetric if $\mathbf{t}=-\mathbf{t}$ \cite{singh2009limits}. To find the mass points of $X$, we discretize the real-line using a grid  with step-size 0.1, and optimize the distribution over the resulting discrete space. Fig.  \ref{fig:delta} shows the resulting achievable rates for SNRs in the range of 0 to 30 dB for various values of $(n_q,\ell,\delta)$. It can be observed that the performance improvements due to the use of higher degree polynomials are more significant at high SNRs. Furthermore, it can be observed that the set of achievable rates only depends on $min(\ell^{n_q},(\ell-1)\delta n_q+1)$. As a result, for instance the achievable rate when $n_q=1,\ell=2,\delta=2$ is the same as that of $n_q=1,\ell=2,\delta=1$ as shown in the figure. So, in this case, using higher degree polynomials does not lead to rate improvements. On the other hand, the achievable rate for $n_q=3,\ell=2,\delta=1$ is lower than that of $n_q=3,\ell=2,\delta=2$ as shown in the figure. So, using higher degree polynomials does lead to rate improvements in this scenario.

 \begin{figure}[t]
\centering 
\vspace{0.05in}
\includegraphics[width=0.8\textwidth]{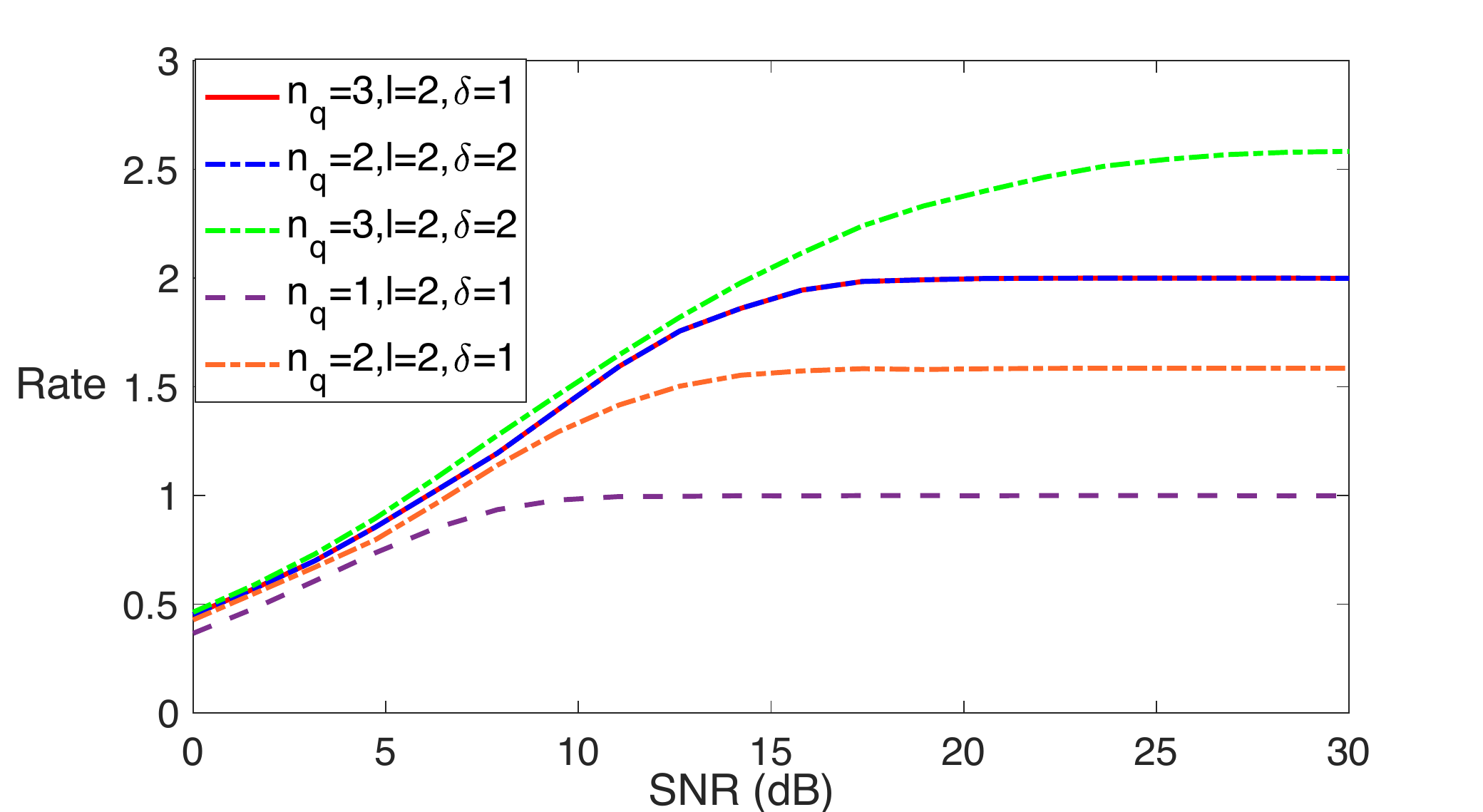}
\vspace{-.15in}
\caption{The set of achievable rates for various values of $(n_q,\ell,\delta)$. 
}
\vspace{-.25in}
\label{fig:delta}
\end{figure}
%\begin{Definition}[\textbf{Code Associated with a Quantizer}]
%Consider a $(\ell,\delta,n_q)$-quantizer characterized by $((f_{a,j}(\cdot))_{j\in [n_q]}, t(1\!\!:\!\!\ell-1,1\!\!:\!\!n_q))$ as described in Definition \ref{def:quant}. Define $f_{a,k,j}(a)\triangleq f_{a,j}(a)-t(k,j), k\in [\ell-1], j\in [n_q]$.
%Let $r_1,r_2,\cdots, r_{\ell  \delta n_q}$ be the sequence of roots of the polynomials $(f_{a,k,j}(\cdot))_{k\in [\ell-1],j\in [n_q]}$ written in an increasing order. That is:
%\begin{align*}
%  &  r_1\leq r_2\leq \cdots \leq r_{\ell \delta n_q}  
%  \\& \forall s\in [\ell \delta n_q] \quad \exists j\in [n_q], k\in [\ell-1]:\quad  f_{a,k,j}(r_s)=0. 
%\end{align*}
%Note that the sequence $(r_s, s\in [\ell \delta n_q])$ may contain repeated roots. The
%the associated code of the quantizer is defined as $\mathcal{C}= (\mathbf{c}_{1}, \mathbf{c}_2,\mathbf{c}_{\ell \delta n_q+1})$, where $\mathbf{c}_{i}= \widehat{W}(1:n_q)$, where $\widehat{W}(1:n_q)$ is the quantizer output given the channel output $Y$ in the interval $[r_{s-1}, r_{s}]$, and we have defined $r_0\triangleq -\infty$. 
%\end{Definition}
 \begin{figure}[!t]
\centering 
\includegraphics[width=0.8\textwidth]{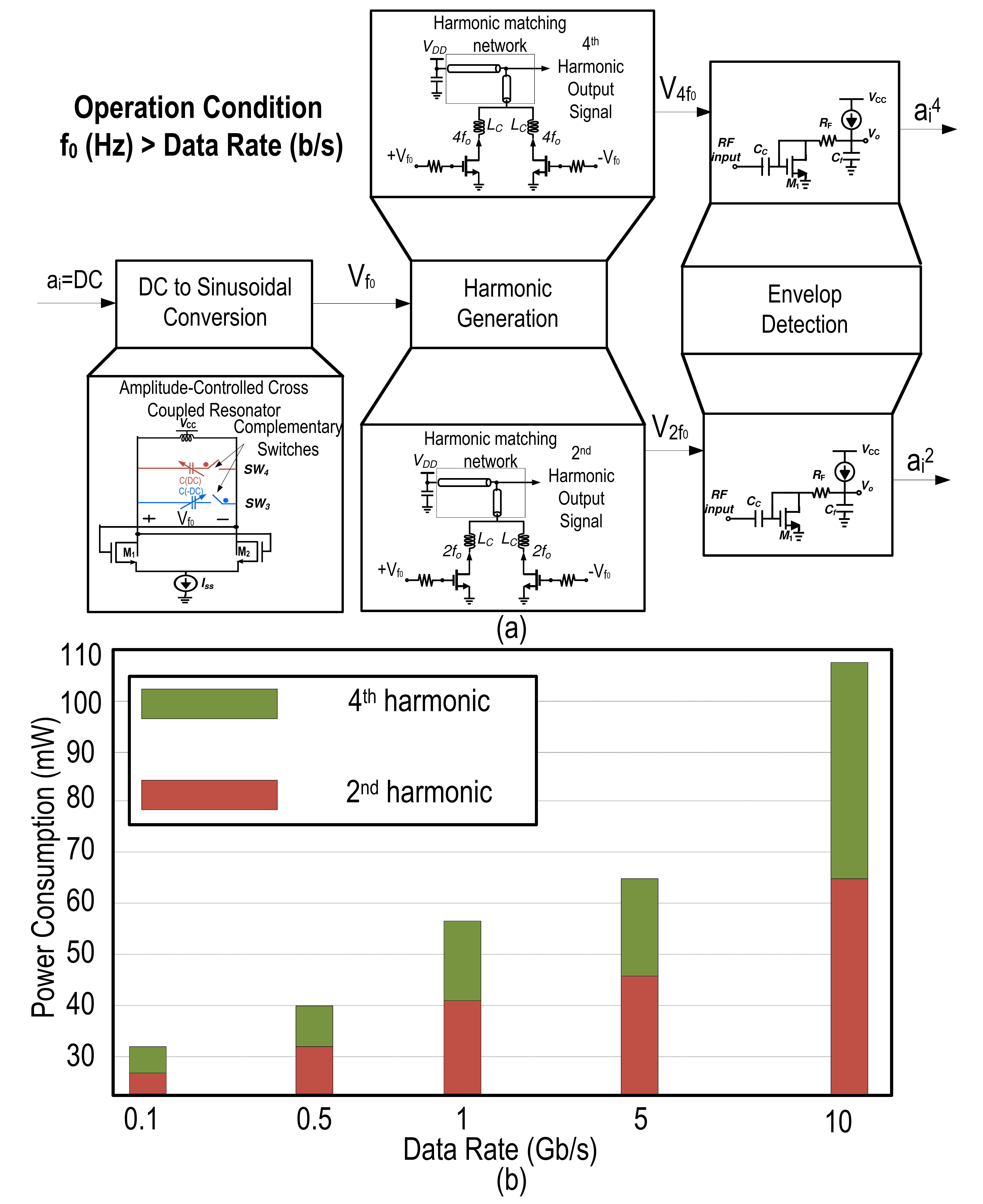}

\caption{(a) The circuit design for the generation of fourth and second order polynomials, (b) the power consumption breakdown of the circuits for generation of equal voltage amplitude (corresponding to 0 dBm power) at the second and fourth harmonics.}
\label{fig:circuit}
\end{figure}

\section{Circuit Design for  Polynomials of Degree Up to Four}
In \cite{shirani2022MIMO}, we considered a single-carrier system, where the baseband input signal is a $sinc(\cdot)$ function and showed the feasibility of implementing quadratic analog operators. The proposed operation involved two steps. In the first step, we used an integrator to transform the signal into a direct current (DC) value. In multi-carrier systems such as orthogonal frequency division multiplexing (OFDM), one can use an analog discrete Fourier transform (DFT) to produce DC signals representing the Fourier coefficients \cite{ganguly2019novel}. As a result, in this section, we assume that we are given a DC signal and our objective is to produce a polynomial function of degree up to four of the input DC value. More precisely, let the input be represented by $X_{DC}$. We wish to produce $\sum_{j=1}^4 b_j X_{DC}^j, b_j\in \mathbb{R}$. In practice, there are two methods to realize the desired polynomials: (i) DC domain nonlinear function synthesis based on the quadratic I-V characteristic of the transistor and increasing the order of polynomial by cascading circuits \cite{DCNON}, (ii) translation of DC values to sinusoidal waveforms, and then generating harmonics of these waveforms with polynomial amplitude which is dependent on the fundamental frequency amplitude. The latter is the method used in  \cite{shirani2022MIMO} to produce quadratic functions. The former has a simpler circuitry; however it can only be used to produce a specific subset of polynomials, and we do not have freedom to choose $b_j, j\in [4]$ arbitrarily. On the other hand, the latter can produce polynomials with arbitrary coefficients through efficient filtering of the undesired harmonic terms. However, it leads to higher power consumption and more complex circuitry. The implementation of this approach requires careful quantification of the non-lienar behavior of transistors which  can be accomplished using the Volterra-Weiner series representation methods \cite{AghasiTHz}.

To explain the proposed construction, let us consider the problem of producing a fourth order polynomial in the form of $f(x)=x^4+x^2$, where $x$ is the  DC input value. It is well-known, that naturally the amplitude level of the fourth harmonic is less than that of the second harmonic. So, a harmonic-centric power optimization is needed to produce the desired polynomial \cite{AghasiTHz}.  Fig. \ref{fig:circuit}(a) shows a circuit design to generate $f(x)=x^4+x^2$. In order to generate equal amplitudes at the second and fourth harmonics, the power gain of the transistors in charge of generating the fourth harmonic should be larger to compensate for the lower harmonic efficiency, leading to an increased power consumption in generating the fourth order term compared to the second order term. We have numerically calculated the power consumption values of the proposed circuit through simulations as shown in Figure \ref{fig:circuit}(b). It can be observed that the ratio of the power consumption for the generation of fourth order term compared to the second order term increases with frequency since the transistor power gain drops at higher frequencies.  These results are based on CMOS 65nm technology. The power consumption can be further improved by transitioning into smaller transistor nodes.

Theorem \ref{th:2} shows that the channel capacity depends on the number of ADCs through $\delta n_q+1$, so that the use of a quadratic analog operator instead of a linear operator ($\delta:1\to 2$) has an equivalent effect on capacity as that of doubling the number of ADCs $n_q$. This fact, along with the power consumptiont values provided in Figure \ref{fig:circuit}(b) justify the use of nonlinear analog operators. It should be noted that the power  consumption is dependent on the circuit configuration, transistor size, and passive quality factors. These simulations serve as a proof-of-concept to justify the effectiveness of the proposed receiver architecture designs. 
\label{sec:cir}

\section{Conclusion}
Application of nonlinear analog operations in SISO receivers was considered. Capacity expressions under various assumptions on the set of implementable nonlinear analog functions were derived.  For systems with one-bit ADCs, it was shown that the capacity is a function of $\delta n_q$, where, $\delta$ is the degree polynomials, and $n_q$ is the number of ADCs. This implies that doubling polynomial degree increases the channel capacity by the same amount as doubling the number of ADCs. Furthermore, circuit-level simulations, using a 65 nm Bulk CMOS technology, were provided to show the implementability of the desired nonlinear analog operators with practical power budgets.

\bibliographystyle{unsrt}
\bibliography{References}

\end{document}